\documentclass[11pt]{article}
\usepackage[colorlinks=false,hidelinks]{hyperref}
\usepackage{amsmath,amssymb,fullpage,times,amsthm,color}
\usepackage[final]{graphicx}

\usepackage{enumitem}

\def\part{{\color{black}{part}}}

\newtheorem*{sandwichability}{Sandwichability}

\newtheorem*{costconsist}{$\gamma$-consistency}

\newtheorem*{rgfc}{Random Graphs of Bounded Cost}

\newtheorem*{tacn}{Thickness and Condition Number}

\newtheorem*{cg}{Cost-geometries}

\newtheorem*{rba}{Rank Based Augmentation}
\newtheorem*{navigability}{Navigability}
\newtheorem*{geometry}{Geometry}
\newtheorem*{optimizer}{Entropic Optimizer}
\newtheorem*{convexity}{Convexity}
\newtheorem*{edgeprofile}{Edge Profile}
\newtheorem*{epentropy}{Edge Profile Entropy}

\newtheorem*{coherence}{Coherence}

\newtheorem*{richness}{Uniform Richness}
\newtheorem*{substrate}{Substrate}
\newtheorem*{reducibility}{Reducibility}

\def\R{\mathrm{I\!R}}

\def\P{\mathbb{P}}

\def\erdos{Erd\H{o}s}
\def\renyi{R\'{e}nyi}
\def\ER{\erdos-\renyi}
\def\ENT{{\sc Ent}}
\def\MU{Thickness}

\newcommand{\sand}{sandwichable}

\newtheorem{theorem}{Theorem}
\newtheorem{proposition}{Proposition}
\newtheorem{definition}{Definition}
\newtheorem{lemma}{Lemma}

\newcommand{\dima}[1]{{\color{black}{#1}}}

\newcommand{\blue}[1]{{\color{black}{#1}}}
\newcommand{\red}[1]{{\color{black}{#1}}}
\newcommand{\mg}[1]{{\color{black}{#1}}}

\def\P{\mathbb{P}}
\newcommand{\reals}{\mathbb{R}}
\newcommand{\nats}{\mathbb{N}}

\begin{document}

\title{Navigability is a Robust Property}

\author{
Dimitris Achlioptas 
\thanks{Research supported by a European Research Council (ERC) Starting Grant (StG-210743) and an Alfred P. Sloan Fellowship.}  
\\ Department of Computer Science\\ University of California, Santa Cruz\\{\tt optas@cs.ucsc.edu}\\ 
\and 
Paris Siminelakis
\thanks{Supported in part by an Onassis Foundation Scholarship.}
 \\ Department of Electrical Engineering\\ Stanford University\\{\tt psimin@stanford.edu}
}

\date{\empty}
 
\maketitle 

\begin{abstract} 
The Small World phenomenon has inspired researchers across a number of fields. A breakthrough in its understanding was made by Kleinberg who introduced Rank Based Augmentation (RBA): add to each vertex independently an arc to a random destination selected from a carefully crafted probability distribution. Kleinberg proved that RBA makes many networks \emph{navigable}, i.e., it allows greedy routing to successfully deliver messages between any two vertices in a polylogarithmic number of steps. We prove that navigability is an inherent property of many random networks, arising without coordination, or even independence assumptions.
\end{abstract}
  
\newpage
\section{Introduction}
 
The Small World phenomenon refers to the fact that there exist short chains of acquaintances between most pairs of people in the world, popularly known as Six Degrees of Separation~\cite{watts2003six}. Milgram's famous 1967 experiment~\cite{Millgram} showed that not only such chains exist, but they can also be found in a decentralized manner. Specifically, each participant in the experiment was handed a letter addressed to a certain person and was told of some general characteristics of the person, including their occupation and location. They were then asked to forward the letter to the individual they knew on a first-name basis who was most likely to know the recipient. Based on the premise that similar individuals have higher chance of knowing each other (homophily),  the participants typically forwarded the message to their contact most similar to the target, a strategy that yielded remarkably short paths for most letters that reached their target (many did not).
 
Kleinberg's groundbreaking work, formulated mathematically the property of finding short-paths in a decentralized manner as \emph{navigability}~\cite{STOC, Nature}. Since then, much progress has been made~\cite{ICM} and the concept of navigability has found  applications in the design of P2P networks~
\cite{Freenet00,Zhang}, data-structures~\cite{aspnes2002fault,Slivkins05} and search algorithms~\cite{Manku04,zeng2005near}. Key to decentralization is shared knowledge in the form of geometry, i.e., shared knowledge of a (distance) function on pairs of vertices (not necessarily satisfying the triangle inequality).  

\begin{geometry}
A \emph{geometry} $(V,d)$ consists of a set of vertices $V$ and a distance function $d:V\times V\to \R_{+}$, where $d(x,y)\geq 0$, $d(x,y)=0$ iff $x=y$,  and  $d(x,y)=d(y,x)$, i.e., $d$ must be a semi-metric.  
\end{geometry}

Given a graph $G(V,E)$ on a geometry $(V,d)$, a \emph{decentralized search algorithm} is any algorithm that given a target vertex $t$ and current vertex $v$ selects the next edge $\{v,u\} \in E$ to cross by only considering the distance of each neighbor $u$ of $v$ to the target $t$, i.e., $d(u,t)$. The allowance of paths of polylogarithmic length in the definition of navigability, below, is motivated by the fact that in any graph with constant degree the diameter is $\Omega(\log(n))$, reflecting an allowance for polynomial loss due to the lack of global information. 

\begin{navigability}
A graph $G(V,E)$ on geometry $(V,d)$ is \red{$d$-\emph{navigable}} if there exists a decentralized search algorithm which given any two $s,t \in V$ will find a path from $s$ to $t$ of length $O\left(\mathrm{poly}(\log |V|)\right)$.
\end{navigability}

In his original work on navigability~\cite{STOC, Nature}, Kleinberg showed that if $G$ is the 2-dimensional grid then adding a single random edge independently to each $v \in V$ results in a navigable graph (with $d$ being the L1 distance on the grid). The distribution for selecting the other endpoint $u$ of each added edge is crucial. Indeed, if it can only depend on $d(v,u)$ and distinct vertices are augmented independently, Kleinberg showed that there is a \emph{unique} suitable distribution, the one in which the probability is proportional to $d(v,u)^{-2}$ (and, more generally, $d(v,u)^{-r}$ for $r$-dimensional lattices). The underlying principle behind Kleinberg's augmentation scheme has by now become known as \emph{Rank Based Augmentation} (RBA)~\cite{NIPS,Liben}.	
\begin{rba}
Given a geometry $(V,d)$, a vertex $v \in V$, and $\ell \ge 0$, let $N_{v}(\ell)$ be the \red{number} of vertices \emph{within} distance $\ell$ from $u$. In RBA, the probability of augmenting $v$ with an edge to any $u \in V$ is
\begin{equation}\label{eq:rba}
P(v,u) \propto \frac{1}{N_{u}\left(d(v,u)\right)} \enspace .
\end{equation}  
\end{rba}  
 
The intuition behind RBA is that navigability is attained because the added edges provide connectivity \emph{across all distance scales}. Concretely, observe that for any partition of the range of the distance function $d$ into intervals, the probability that the (distance of the) other endpoint of an added edge lies in a given interval is independent of the interval. Therefore, by partitioning the range of $d$ into $O(\log n)$ intervals we see that whatever the current vertex $v$ is, there is always $\Omega((\log n)^{-1})$ probability that its long-range edge is to a vertex at a distance at the same ``scale" as the target. Of course, that alone is not enough. In order to shrink the distance to the target by a constant factor, we also need the long-range edge to have reasonable probability to go ``in the right direction", something which is effortlessly true in regular lattices for any finite dimension. In subsequent work~\cite{NIPS}, aiming to provide rigorous results for graphs beyond lattices, Kleinberg showed that the geometric conditions needed for RBA to achieve navigability are satisfied by the geometries induced by \emph{set-systems} satisfying certain conditions when the distance between two vertices is the  size of the smallest set (homophily) containing both (see Definition~\ref{ss:def} in Section~\ref{sec:ssarecoherent}).  
 
Another canonical setting for achieving navigability by RBA is when the distance function $d$ is the shortest-path metric of a connected graph $G_{0}(V, E_{0})$ with large diameter $\Theta(\mathrm{poly}(n))$. In that setting, if $E_{d}$ is the random set of edges added through RBA, the question is whether the (random) graph $G(V,E_{0}\cup E_{d})$ is \red{$d$-navigable}. Works of Slivkins~\cite{Slivkins05} and Fraigniaud et al.~\cite{Lower} have shown the existence of a threshold, below which navigability is attainable and above which (in the worst case) it is not attainable, in terms of the \emph{doubling dimension} of the shortest path metric of $G_0$. Roughly speaking, the doubling dimension corresponds to the logarithm of the possible directions that one might need to search, and the threshold occurs when it crosses $\Theta(\log \log n)$. Thus, we see that even when $d$ is a (shortest path) metric, very significant additional constraints on $d$ need to be imposed. 

The remarkable success of RBA in conferring navigability rests crucially on its  \emph{perfect adaptation} to the underlying geometry. This adaptation, though, not only requires perfect independence and identical behavior of all vertices, but also a very specific, indeed unique, functional form for the probability distribution of edge formation. This exact fine tuning renders RBA unnatural greatly weakening its plausibility. \blue{Our goal in this paper is to demonstrate that navigability is in fact a robust property of networks that emerges from very basic considerations without adaptation, coordination, or even independence assumptions.}   

\section{Our Contribution}

\blue{As mentioned, at the foundation of navigability lies shared knowledge in the form of geometry.} Our starting premise is that geometry imposes \emph{global} constraints on the set of feasible networks. Most obviously, in a physical network where edges (wire, roads) correspond to a resource (copper, concrete) there is typically an upper bound on how much can be invested to create the network. \blue{More generally, cost may represent a number of different notions that distinguish between edges.} We formalize this intuition by (i) allowing different edges to have arbitrary costs, i.e., without imposing any constraints on the cost structure, and (ii) taking as input an upper bound on the \emph{total cost} of feasible graphs, i.e., a budget. We remain fully agnostic in all other respects, i.e., we study the \emph{uniform} measure on all graphs satisfying the budget constraint. \blue{So, for example, if all edges have unit cost we recover the classic \ER\ $G(n,m)$ random graphs (except now $m$ is a random variable, sharply concentrated just below the budget.)} 
 
As one can imagine, the set of all graphs feasible within a given budget may contain wildly different elements. Our capacity to study the uniform measure on such graphs comes from a very recent general theorem we developed in~\cite{product}, of which this work is the first application. At a high level, the main theorem of~\cite{product} asserts that if a subset $S$ of the set of all undirected simple graphs $\mathcal{G}_{n}$ on $n$ vertices is sufficiently symmetric, then the uniform measure on $S$ can be well-approximated by a product measure on the edges, i.e., a measure where each edge is included independently with different edges potentially having different probabilities. Formally, a product measure on $\mathcal{G}_{n}$ is specified succinctly by a symmetric matrix $\mathbf{Q} \in [0,1]^{n\times n}$ of probabilities where $Q_{ii}=0$ for $i\in[n]$. We denote by $G(n,\mathbf{Q})$ the measure in which possible edge $\{i,j\}$ is included independently with probability $Q_{ij}=Q_{ji}$. The main result of~\cite{product} then allows one to approximate the uniform measure by a product measure in the following very strong sense.

\begin{sandwichability}
The uniform measure $U(S)$ over an arbitrary set of graphs $S \subseteq \mathcal{G}_{n}$ is \emph{$(\epsilon,\delta)$-\sand}\ if there exists a $n\times n$ symmetric matrix $\mathbf{Q}$ such that the two distributions $G^{\pm} \sim G(n,(1\pm\epsilon)\mathbf{Q})$, and the distribution $G \sim U(S)$ can be coupled so  that $G^{-} \subseteq G \subseteq G^{+}$ with probability at least $1-\delta$.
\end{sandwichability}

As discussed above, navigability requires some degree of structure in the underlying geometry. It is from this structure that we will extract the symmetry needed to apply the theorem of~\cite{product} and derive a product form approximation for graphs with a bounded total cost. Armed with such an approximation, establishing navigability becomes dramatically easier, allowing us to demonstrate its robustness and ubiquity. \blue{Roughly speaking, we isolate three ingredients that suffice for navigability on a geometry $(V,d)$:
\begin{itemize}
\item
A \emph{substrate} of connections between nearby points on $V$, allowing the walk to never get stuck. 
\item
Some degree of \emph{coherence} of the distance function $d$.
\item
Sufficient, and sufficiently uniform, edge density across all distance scales. 
\end{itemize} 
The first two ingredients are generalizations of existing work and, as we will see, fully compatible with RBA. The third ingredient is also motivated by the RBA viewpoint, but we will prove that it can be achieved in far more-light handed, and thus natural, manner than RBA. Moreover, in the course of doing so, we will give it a very natural \emph{economical} interpretation as the \emph{cost of indexing}.} 
 
\subsection{The two Basic Requirements and a Unifying Framework for RBA}

\begin{substrate}
A set of edges $E_{0}$ forms a \emph{substrate} for a geometry $(V,d)$, if for every $(s,t) \in V \times V$ with $s\neq t$, there is at least one vertex $v$ such that $\{s,v\} \in E_{0}$ and $d(v,t) \le d(s,t)-1$.
\end{substrate}
The existence of the substrate implies that (very slow) travel between any two vertices is possible, so that a decentralized algorithm never gets trivially stuck. \medskip

\blue{Coherence is a notion that comes with an associated scale factor $\gamma>1$.  Specifically, given a geometry $(V,d)$ we will refer to the vertices whose distance from a given vertex $v \in V$ lie in the interval $(\gamma^{k-1},\gamma^{k}]$ as the vertices in the $k$-th (distance) $\gamma$-scale from $v$. Also, for a fixed $\lambda < 1$ and any target vertex $t\neq v$, we will say that a vertex $u$ is $t$-helpful to $v$ if $d(v,u) \leq \gamma^{k}$ ($u$ is within the same $\gamma$-scale as $t$), and $d(u,t) < \lambda d(v,t)$ (reduces the distance by a constant).} 
 
\blue{\begin{coherence} 
 Let $K = \lceil \log_{\gamma}|V|\rceil$. A geometry $(V,d)$ is $\gamma$-\emph{coherent} if there is $\lambda < 1$ such that for all $v \in V$:\smallskip 

\noindent -- For all $k \in [K]$, the number of vertices in the $k$-th distance scale from $v$ is $P_k(v) = \Theta(\gamma^k)$.\smallskip

\noindent -- For all $t\neq v$, a constant fraction of the vertices whose distance scale from $v$ is no greater than the distance scale of $t$ are $t$-helpful to $v$.\smallskip
\end{coherence}

The two conditions above endow the, otherwise arbitrary, semi-metric $d$ with sufficient regularity and consistency to guide the search. Although our definition of coherence is far more general, in order to convey intuition about the two conditions,  think for a moment of $V$ as a set of points in Euclidean space. The first condition guarantees that there are no ``holes", as the variance in the density of points is bounded in every distance scale. The second condition guarantees that around any vertex $v$ the density of points does not change much depending on the direction (target vertex $t$) and distance scale. Besides those two conditions, we make \emph{no further} assumptions on $d$ and, in particular, we do \emph{not} assume the triangle inequality.} 

Coherent geometries allow us to provide a unified treatment of navigability since they encompass finite-dimensional lattices, hierarchical models, any vertex transitive graph with bounded doubling dimension and, as we prove in Section~\ref{sec:ssarecoherent}, Kleinberg's set systems~\cite{NIPS} (see Definition~\ref{ss:def}).
 
\begin{theorem}\label{coh-set}
Every set system satisfying the conditions of~\cite{NIPS} is a $\gamma$-coherent geometry for some $\gamma>1$.
\end{theorem} 

\begin{theorem}\label{thm:rba} 
Let $(V,d)$ be any $\gamma$-coherent geometry and let $E_{0}$ be any substrate for it. If $E_{d}$ is the (random) set of edges obtained by applying RBA to $(V,d)$, then the graph $G(V,E_{0}\cup E_{d})$ is $d$-navigable w.h.p.\footnote{Throughout the paper, all asymptotics will be with respect to the number of vertices $|V|=n$. Thus, with high probability will always mean ``with probability that tends to 1 as $n \to \infty$."}
\end{theorem}
  
Theorem~\ref{thm:rba} subsumes and unifies a number of previous positive results on RBA-induced navigability. Our main contribution, though, lies in showing that given a substrate and coherence, navigability can emerge without coordination from the interplay of cost and geometry.

\subsection{Navigability from Organic Growth}

As mentioned earlier, the success of RBA stems from the fact that the edge-creation mechanism is \emph{perfectly} adapted to the underlying geometry so as to induce navigability. In contrast, we will not even specify an edge-creation mechanism, but rather focus only on the set of graphs feasible with a given budget. Our requirement is merely that the cost function is \emph{informed} by the geometry. 
\begin{costconsist}  
\blue{Given a $\gamma$-coherent geometry $(V,d)$, a cost function $c: V \times V \to \reals$ is $\gamma$-\emph{consistent} if $c$ takes the same value $c_{k}$ for every edge $\{u,v\}$ such that $d(u,v) \in (\gamma^{k-1},\gamma^{k}]$.}
\end{costconsist} 
 
In particular, note that  we make no requirements of the values $\{c_k\}$, not even a rudimentary one, such as being increasing in $k$. All that $\gamma$-consistency entails is that the partition of edges according to cost is a coarsening of the partition of the edges by $\gamma$-scale. In fact, even this requirement can be weakened significantly, as long as some correlation between the two partitions remains, but it is technically much simpler to assume $\gamma$-consistency as it simplifies the exposition greatly. One can think of consistency as limited sensitivity with respect to distance. As an example, it means that making friends with the people next door might be more likely than making friends with other people on the same floor, and that making friends with people on the same floor is more likely than making friends with people in a different floor, but it does not really matter which floor.  
\red{\begin{cg}
We say that $\Gamma = \Gamma(V,d,c)$ is a \emph{coherent cost-geometry} if there exists $\gamma > 1$ such that $(V,d)$ is a $\gamma$-coherent geometry and $c$ is $\gamma$-consistent cost function.
\end{cg}}

\begin{rgfc}
Given a coherent cost-geometry $\Gamma(V,d,c)$ and a real number $B \ge 0$, let $\red{G_{\Gamma}(B)} =\{E\subseteq V\times V: \frac{1}{n}\sum_{e\in E} \red{c(e)} \leq B \}$, i.e.,  $\red{G_{\Gamma}(B)}$ is the set of all graphs (edge sets) on $V$ with total cost at most $Bn$. A uniformly random element of $\red{G_{\Gamma}(B)}$ will be denoted as \red{$E_{\Gamma}= E_{\Gamma}(B)$}. 
\end{rgfc}

Applying the main theorem of~\cite{product}, in Section~\ref{sec:product} we will prove that random graphs of bounded cost (on a consistent cost-geometry) have a product measure approximation, in the following sense. 
\begin{theorem}\label{partakariola}
Given a coherent cost-geometry $\Gamma$, there exist a a unique function $\lambda(B)\geq 0$ and  constant $B_{0}(\Gamma)>0$  such that for every $B\geq \red{B_{0}(\Gamma)}$ the uniform measure on $G_{\Gamma}(B)$ is $(\delta, \epsilon)$-sandwichable by the product measure 
in which the probability of every edge with cost $c_k$ is
\begin{equation}\label{ohyeah}
\frac{1}{1+\exp(\lambda(B) c_{k})} \enspace ,
\end{equation}
and $(\delta, \epsilon) = \red{\left(\sqrt{\frac{24}{\log \red{|V|}}}, \ 2 \red{|V|}^{- 5K}\right)}$. The number $\lambda(B)>0$ can be explicitly defined in terms of $\{c_{k}\}$.
\end{theorem}
The regularizer $\lambda=\lambda(B)$ in Theorem~\ref{partakariola} corresponds to the derivative of entropy with respect to the budget $B$ (energy), i.e., is an inverse temperature, and depends on $B$ in a smooth one-to-one manner. 

\blue{Theorem~\ref{partakariola} will give us a great amount of access to the uniform measure on $G_{\Gamma}(B)$. In particular, the upper approximation \mg{$G(|V|,(1+\epsilon)\mathbf{Q})$} will allow us to bound the total number of edges present in a typical graph, establishing sparsity for all sufficiently small budgets. On the other hand, the lower approximation \mg{$G(|V|,(1-\epsilon)\mathbf{Q})$}  will allow us to establish a lower bound on the number of edges incident to each vertex of each distance scale. Combined with the spatial uniformity afforded by independence, this will allow us to prove that navigability emerges as soon as the total number of edges within each scale is large enough, establishing navigability for all sufficiently large budgets.}

\setlist[description]{font=\sffamily\normalfont}

\begin{theorem}\label{thm-general}
For every coherent \red{cost-geometry} $\Gamma(V,d,c)$, where $|V|=n$, there exist numbers $B^{\pm}$ such that if $E_{\Gamma}$ is a uniformly random element of $G_{\Gamma}(B)$ then:
\begin{description}
\item[--]
For all $B\leq B^{\red{+}}$, w.h.p.\ $|E_{\Gamma}|=O(n\cdot \mathrm{poly}(\log n))$. \hfill (Sparsity) 
\item[--]
For all $B\geq B^{\red{-}}$, for any substrate $E_{0}$, w.h.p.\ the graph $G(V,E_{0}\cup E_{\Gamma})$ is $d$-navigable. \hfill (Navigability)
\end{description} 
\end{theorem}

Note that Theorem~\ref{thm-general} shows that navigability arises eventually, i.e., for all $B\geq B^{-}$, without \emph{any} further assumptions on the cost function or geometry. 
The caveat, if we think of $B$ as increasing from 0, is that by the time there are enough edges across all distance scales, i.e., $B \ge B^{-}$, the total number of edges may be much greater than linear. This is because for an arbitrary cost structure $\{c_k\}$, by the time the  ``slowest growing" distance scale has the required number of edges, the other scales may be replete with edges, possibly many more than the required $\Omega(n/\mathrm{poly}\log(n))$. This is reflected in the ordering between $B^{-}$ and $B^{+}$ that determines whether the sparsity and navigability regimes are overlapping. In particular, we would like $B^{-} \leq B^{+}$ and, ideally, the ratio $R = B^{+}/B^{-}>0$ to be large. Whether this is the case or not depends precisely on the degree of adaptation of the cost-structure $\{c_k\}$ to the geometry as we discuss next.


 
\subsection{Navigability as a Reflection of the Cost of  Indexing}

Recall that for every vertex $v$ in a $\gamma$-coherent geometry and for every distance scale $k \in [K]$, the number of vertices whose distance from $v$ is in the $k$-th $\gamma$-distance scale is $P_k(v) =\Theta(\gamma^{k})$. At the same time, \eqref{ohyeah} asserts that the probability of each edge is exponentially small in its cost. Thus, reconciling sparsity with navigability boils down to balancing these two factors. We will exhibit a class of cost functions that (i) have an intuitive interpretation as the \emph{cost of indexing}, (ii) achieve a ratio $R =  B^{+}/B^{-}>0 $ that \emph{grows} with $n$, i.e., a very wide range of budgets for which we have both navigability and sparsity, and (iii) recover RBA as a special case corresponding to a particular budget choice. 
 
Consider a vertex $v$ that needs to forward a message to a neighbor $u$ at the $k$-th distance scale. To do so, $v$ needs to distinguish $u$ among all other $P_k(v)$ vertices at the $k$-th distance scale, i.e., $v$ needs to be able to \emph{index} into that scale. To do so, it is natural to assert that $v$ must incur a cost of $\Theta(\log_2 P_k(v)) = \Theta(k)$ (due to coherence) bits to store the unique ID of $u$ among the other members of its equivalence class (in the eyes of $v$). Motivated by this we consider cost functions where for some $\alpha > 0$,
\[
c^*_k = \frac{1}{\alpha}k \enspace .
\]


\begin{theorem}\label{thm-indexing}
For any coherent cost-geometry $\Gamma(V, d, c^{*})$, where $|V|=n$, there exist $B^{-} \le B^{+}$ such that: 
\begin{enumerate}
\item\label{window} 
$B^{+} /B^{-} = \omega(\mathrm{poly}\log n)$.
\item 
For all $B \in [B^{-}, \ B^{+}]$, w.h.p.:
\begin{itemize}  
\item
$|E_{\Gamma}(B)|= O(n \;\mathrm{poly}\log n))$.
\item
The graph $G\left(V, E_{0} \cup E_{\Gamma}(B)\right)$ is $d$-navigable.
\end{itemize}
\item
There exists $B_{a} \in [B^{-}, \ B^{+}]$ such that in the approximation of $E_{\Gamma}$ by \mg{$G(|V|,\mathbf{Q})$}, for every $\{u,v\} \in E$, $\mathbf{Q}_{ij}^{*} = \Theta(N_{u}(d(u,v))^{-1})$, i.e, Rank Based Augmentation is approximately recovered.
\end{enumerate}
\end{theorem}


Part~\ref{window} of Theorem~\ref{thm-indexing} is equivalent to a scaling window of $\Theta(\frac{\log\log n}{\log n})$ for the exponent $\lambda$, within which navigability holds with poly-logarithmic average degree. This corroborates Kleinberg's work that gave a unique exponent of $\beta=-1$ in the context of RBA for the scaling~\eqref{eq:rba} of probability. Nevertheless, under our framework this vanishing window for the highly sensitive paramater $\lambda$ produces a \emph{diverging} range for values of $B$, explaining the purported fragility of RBA to looking at perturbations in the \emph{wrong scale}. In fact, we can use this feature of our model to provide the first theoretical explanation for the discrepancy between theoretical results and empirical evidence~\cite{Adamic,Liben,Clauset03} showing that real networks exhibit an exponent $\hat{\beta} \approx 0.8 < 1$. In our setting, exponents smaller than 1 correspond to \emph{higher} average degree and thus we can attribute this discrepancy to finite size effects (finite $n$) and the densification~\cite{leskovec} of networks.

\section{Deriving a Product Measure Approximation: Proof of Theorem~\ref{partakariola}}\label{sec:product}

We start with some definitions that will allow us to state the main theorem of~\cite{product}. A set of graphs $S\subseteq \mathcal{G}_{n}$  is symmetric with respect to a partition $\mathcal{P}$ of the set of all possible $\binom{n}{2}$ edges, if the characteristic function of $S$ depends only on the number of edges from each part of $\mathcal{P}$ but not on which edges.

\begin{edgeprofile} Given a partition $\mathcal{P}=(\mathcal{P}_{1},\ldots,\mathcal{P}_{K})$ of the set of all possible $\binom{n}{2}$ edges, for a set of edges $E \in \mathcal{G}_{n}$ and for each $k \in [K]$, let $m_{k}(E)$ denote the number of edges in $E$ from $\mathcal{P}_{k}$. The \emph{edge profile} of $E$  is $\mathbf{m}(E) := (m_{1}(E),\ldots,m_{K}(E))$. 
\end{edgeprofile}

\noindent We denote the image of a symmetric set $S$ under the edge-profile as $\mathbf{m}(S)$. As before let $P_{k}=|\mathcal{P}_{k}|=\frac{1}{2}\sum_{u\in V}P_{k}(u)$ be the total number of edges in \part\ $k$ of partition $\mathcal{P}$.
\vspace{-0.4cm}

\begin{epentropy}
Given an edge profile $\mathbf{v} = (v_1,\ldots,v_k)$ the entropy of $\mathbf{v}$ is
$\displaystyle{\text{\ENT}(\mathbf{v}) = \sum_{k=1}^{K}\log \binom{P_{k}}{v_{k}}}$.
\end{epentropy}

\noindent The edge-profile entropy is used to express the number of graphs  with a particular edge profile $\mathbf{v}$ as $\exp(\text{\ENT}(\mathbf{v}))$. Given any symmetric set $S\subseteq \mathcal{G}_{n}$, the probability of observing an edge profile $\mathbf{v}$ when sampling an element uniformly at random from $S$ is then given by    $
\mathbf{P}_{S}(\mathbf{v}) =  \frac{1}{|S|}  e^{\text{\ENT}(\mathbf{v})}$. Thus, in order to analyze the distribution of a random edge-profile, and consequently of a random element of $G_{c}(n,B)$, we are going to exploit analytic properties of the entropy on the set of feasible edge profiles $\mathbf{m}(S)$.

\begin{convexity}\label{def-convex}
Let $\mathrm{Conv}(A)$ denote the convex hull of a set $A$. Say that a $\mathcal{P}$-symmetric set $S\subseteq \mathcal{G}_{n}$ is convex iff the convex hull of $\mathbf{m}(S)$ contains no new integer points, i.e., if $\mathrm{Conv}(\mathbf{m}(S))\cap \nats^{k} = \mathbf{m}(S)$.
\end{convexity}


\begin{optimizer} Given a symmetric set $S$, 
let $\mathbf{m}^{*}=\mathbf{m}^{*}(S)\in \R^{k}$ be the unique solution to
\begin{equation}
\label{eq:optimization}
\max_{\mathbf{v} \in \mathrm{Conv}(\mathbf{m}(S))} -\sum_{k=1}^{K}\left[v_{k}\log\left(\frac{v_{k}}{P_{k}} \right) +(P_{k}-v_{k})\log\left(\frac{P_{k}-v_{k}}{P_{k}} \right)\right].
\end{equation}
\end{optimizer}

Given the maximizer $\mathbf{m}^{*}(S)$, the matrix $\mathbf{Q}^{*}=\mathbf{Q}^{*}(S)$  is given by letting for all $k\in[K]$ the probability of an edge $e\in \mathcal{P}_{k}$ be $Q^{*}_{e} := m^{*}_{k}/P_{k}$. To state the theorem,  we need  the following parameters that quantify the concentration of the uniform measure around its mode.

\begin{tacn}
Given a partition $\mathcal{P}$ and a $\mathcal{P}$-symmetric  set $S$, we define
\begin{align}
\text{\emph{\MU:}}\qquad 				& \mu  		=  \mu(S) 	 =	 \min_{k\in[K]} \min\{m^{*}_{k}, P_{k}-m^{*}_{k}\} \label{def:thick} \qquad \qquad \qquad\\
\text{\emph{Condition number:}} \qquad 	& \tau 	= \tau(S)  =    \frac{5K\log n}{\mu(S)} \label{def:cond}
\end{align}
\end{tacn} 

We now state the main theorem employed in the proof. 

\begin{theorem}[\cite{product}]\label{thm-sandwich}
Let $\mathcal{P}$ be any  edge-partition and let $S$ be any  $\mathcal{P}$-symmetric convex set. For every $\epsilon>\sqrt{12\tau(S)}$, the uniform measure over $S$ is $(\epsilon,\delta)$-\sand\  for $\delta =2 \exp\left[ - \mu(S)\left(\frac{\epsilon^{2}}{12}-\tau(S) \right)\right]$.
\end{theorem}

In our  setting, $S$ is the set $G_{\Gamma}(B):=\{E\subset V\times V: \frac{1}{n}\sum_{e\in E}c_{e} \leq B\}$ of graphs with bounded average cost and $\mathcal{P}$ is the partition induced by the coherent cost function $c$. The set $\mathbf{m}(S)$ is then given by $\mathbf{m}(S)=\{ \mathbf{v}\in \nats^{k}: \frac{1}{n}\sum_{k=1}^{K}c_{k} v_{k} \leq B \}$. Hence, it is easy to see that $G_{\Gamma}(B)$ is convex and symmetric, according to the previous definition,  for all values of $B$. To prove Theorem \ref{partakariola}, we need to find:
\begin{itemize}
\item[(i)]an analytic expression for the vector $\mathbf{m}^{*}$ as a function of $B$
\item[(ii)]the range of values of $B$ for which applying Theorem \ref{thm-sandwich} gives high probability bounds.
\end{itemize}
\subsection{Finding the Entropic Optimal Edge Profile}
 We start by introducing  a slight reparametrization  in terms of the average-degree profile. For an edge set $E$, define the vector $\mathbf{a}(E):=\mathbf{m}(E)/n$, where as before $\mathbf{m}$ is the edge-profile. In the same spirit, let $p_{k}=P_{k}/n$ denote the average number of edges in part (scale) $k$. Using this parametrization and by explicitly writing $\mathrm{Conv}(\mathbf{a}(S))$, we can equivalently express the optimization problem \eqref{eq:optimization} as:
\begin{eqnarray*}
\max_{\mathbf{a}} \  H(\mathbf{a}) &= &  \ - \sum_{k=1}^{K} 
  \left[(p_{k}-a_{k})\log(p_{k}-a_{k})+a_{k}\log(a_{k})\right]\\
     \mbox{subject to} \ & \ & \sum_{k=1}^{K}a_{k}c_{k} \leq  B\\
     \ \ & \ &   0\leq a_{k} \leq p_{k}, \qquad \forall k\in[K] \enspace .
\end{eqnarray*}
We will refer to the above  optimization problem as $(\Lambda)$ and to its solution as $\mathbf{a}^{*}=\mathbf{a}^{*}(B)$. Towards obtaining an analytic expression for $\mathbf{a}^{*}$, we first show that no coordinate $k\in[K]$ lies on the natural boundary $\{ 0, p_{k}\}$. 
\begin{lemma}
\label{lm:interior}
The optimal profile $\mathbf{a}^{*} \in 
\mathcal{D}(B):=\{\mathbf{a}\in (0,p_{1})\times \ldots \times (0,p_{K}): \sum_{k}^{K}c_{k}a_{k} \leq B \}$.
\end{lemma}
\begin{proof} We prove the lemma by contradiction. We show that if $\mathbf{a}^{*}$ is a solution of $(\Lambda)$ such that $\mathbf{a}^{*}\notin  \mathcal{D}$, then there is an $\hat{\mathbf{a}}^{*}\in \mathcal{D}$ for which objective function $f$ takes a higher value. Specifically, for $\epsilon>0$ assume that there are indices $1 \leq i,j \leq K$ such that $a^{*}_{i}=0$ and $a^{*}_{j}>\delta(\epsilon)$\footnote{For any nontrivial values of $B$ such an index can always be found.}, where $\delta(\epsilon)=\epsilon \, c_i/c_j$. Define  $\hat{\mathbf{a}}^{*}(\epsilon)=(a^{*}_{1},\ldots,a^{*}_{i}+\epsilon,\ldots, a^{*}_{j}-\delta(\epsilon),\ldots,a^{*}_{K})$. If $h(\epsilon)=H(\hat{\mathbf{a}}^{*})-H(\mathbf{a}^{*})$ is the difference in the objective function between the assumed optimal $\mathbf{a}^{*}$ and the perturbation $\hat{\mathbf{a}}^{*}$, then
\[
h'(\epsilon)=-\log(\epsilon)+\log(p_{i}-a_{i}-\epsilon)+\frac{c_{i}}{c_{j}}\left(\log(a_{j}-\delta(\epsilon))-\log(p_{j}-a_{j}+\delta (\epsilon )) \right) \enspace .
\]
Observe, that $\lim_{\epsilon\to 0}h^{'}(\epsilon)=+\infty$, since we have assumed that $a^{*}_{j}>0$. This shows that every maximizer satisfies $\mathbf{a}^{*}>0$. The same argument establishes that $a^{*}_{k}<p_{k}$ for all $k\in[K]$. Combining the two statements we get that any maximizer belongs in $\mathcal{D}$.
\end{proof}

As a consequence, since they are  inactive at the optimum, we can omit separable inequalities from the formulation.  Further, define $\bar{B}:=\frac{1}{2}\sum_{k=1}^{K}p_{k}c_{k}$ the average cost of the solution to the unconstrained version of $(\Lambda)$, i.e., where $\bar{a}_{k}:=p_{k}/2$. If $B> \bar{B}$ then the absolute maximum entropic point $\bar{\mathbf{a}}$ is still in $\mathcal{D}(B)$ and thus the solution will be always $a^{*}_{k}=\bar{a}_{k}$ for every such $B$.

\begin{lemma}
\label{lm:unique} There is a unique function $\lambda(B)$ that is one-to-one for all $0\leq B\leq \bar{B}$ and $\lambda(B)=0$ for all $B\geq \bar{B}$, such that the unique solution of  $(\Lambda)$ is given by: 
\begin{equation}
\label{eq:optimizer}
a^{*}_{k}(B)=\frac{p_{k}}{1+\exp{[\lambda(B) \cdot c_{k}]}}, \ \forall k\in [K] \enspace .
\end{equation}
\end{lemma}
\begin{proof}
Uniqueness of the solution follows easily from convexity of the domain and concavity of the objective function. Further, by Lemma~\ref{lm:interior}, we can reduce the optimization problem $(\Lambda)$ to the following:
\begin{eqnarray*}
\max_{\mathbf{a}} &- &  \sum_{k=1}^{K}\left[(p_{k}-a_{k})\log(p_{k}-a_{k})+a_{k}\log(a_{k})\right]\\
     \mbox{subject to} \ & \ & \sum_{k=1}^{K}a_{k}c_{k} \leq  B \enspace .
\end{eqnarray*}
To obtain an analytical solution, we form the Lagrangian of the reduced problem
\[
L(\mathbf{a},\lambda)=-   \sum_{k=1}^{K}\left[(p_{k}-a_{k})\log(p_{k}-a_{k})+a_{k}\log(a_{k})\right]+\lambda\left(B-\sum_{k=1}^{K}a_{k}c_{k} \right) \enspace .
\]
with the additional constraint that $\lambda\geq 0$. The Karush-Kuhn-Tacker conditions read 
\begin{eqnarray}
\frac{\partial L}{\partial a_{k}}=0 &\Longleftrightarrow &\log\left(\frac{a_{k}}{p_{k}-a_{k}}\right)=-\lambda c_{k} \label{eq:first}\\
\frac{\partial L}{\partial \lambda}=0 &\Longleftrightarrow &\sum_{k=1}^{K}a_{k}c_{k} =B  \label{eq:constraint}\enspace .
\end{eqnarray}
Solving the first equation for $a_{k}(\lambda)$ we get
\[
a^{*}_{k}=\frac{p_{k}}{1+\exp (\lambda c_{k})} \enspace ,
\]
Substituting this expression in \eqref{eq:constraint}, we get the following function of $\lambda$:
\begin{equation}
\label{def:costfun}
g(\lambda)=\sum_{k=1}^{K}c_{k} \cdot \frac{p_{k}}{1+\exp (\lambda c_{k})}
\end{equation}
and the second constraint can now be written as $g(\lambda)=B$. The domain of  $g$ is the set of non-negative numbers on which $g$ is continuous and infinitely differentiable. Under positive costs $\{c_{k}\}$, it is easy to see that $g'(\lambda)<0$ for all $B<\bar{B}$ , hence, $g$ is strictly decreasing in the interval $[0,\infty)$ and $g(0)=\bar{B}$. Thus, $g:[0,\infty) \to [0,\bar{B}]$ is 1-to-1 and thus invertible. This  means that every budget in $[0,\bar{B}]$ is feasible and that for each such budget there is a unique $\lambda(B):=g^{-1}(B)$. For $B\geq \bar{B}$, $\lambda(B)=0$. Therefore, we conclude that the maximizer is always unique for any feasible $B$ and implicitly given by $g(\lambda)=B$.
\end{proof}

\subsection{Thickness $\mu(B)$ of $G_{\Gamma}(B)$ and Sandwiching}
Our next step is to use the analytical solution to the optimization problem to instantiate the thickness parameter $\mu$ defined in \eqref{def:thick}.  Using \eqref{eq:optimizer}, we can write:
\begin{equation}
\mu(B)=\min_{k\in[K]} m_{k}^{*} = n \cdot \min_{k \in [K]}  \frac{p_{k}}{1+\exp{[\lambda(B)c_{k}]}}
\end{equation}
where we have used the facts that  that $a^{*}_{k}= m_{k}^{*}/n$ and $a^{*}_{k}(B)\leq 1/2 \Rightarrow m_{k}^{*}\leq P_{k} - m_{k}^{*}$. To get a more convenient expression, since $0<c_{k}<\infty$ we can write the cost as $
c_{k}=\frac{1}{\beta_{k}}\log(p_{k})
$ 
where $0<\beta_{k}<\infty$ when  $p_{k}\geq 1$. Thus, approximately\footnote{When the approximation does not hold it means that $\mu(B)=\Omega(n)$ which trivially satisfies all the requirements we need for ``sandwiching" and navigability.} for large $p_{k}$ (eq. $k$) we have $\mu(B)\approx n \cdot \min_{k\in[K]}\left[ p_{k}^{1-\lambda(B)/\beta_{k}}\right]$. Theorem \ref{thm-sandwich}, gives strong (non-constant) probability bounds as long as $\tau(B)\ll 1$. For our purposes we are going to consider that the maximum $\tau(B)$ (respectively minimum $B$) that we allow is $\tau_{0}=\log^{-1}(n)$ (respectively $B_{0}$). Substituting  the above expression for $\mu(B)$ in \eqref{def:cond},  we get that the condition $\tau \leq \tau_{0}$ can be rewritten as $\lambda(B) \leq \lambda_{0}$, where
\begin{equation}
\lambda_{0}=\lambda_{0}(\{p_{k}\},\{\beta_{k}\}):= \min_{k\in[K]}\left[\log\left(\frac{n \log p_{k}}{5K\log^{2}(n)} \right)\frac{\beta_{k}}{\log p_{k}}\right] \enspace .
\end{equation}
Using the function $g(\lambda)$ defined in \eqref{def:costfun}, we can express this constraint as $B\geq B_{0}:=g(\lambda_{0})$. 
 
To conclude the proof of Theorem~\ref{partakariola} we see that $\mu(B) \geq 5K\log^{2}(n)$ and $\tau(B)\leq \frac{1}{\log(n)}$, for all $B\geq B_{0}$. Applying Theorem \ref{thm-sandwich}, for $\epsilon_{0}= \sqrt{\frac{24}{\log n}}$ that is greater than $\sqrt{12\tau_{0}}$, we get that $\delta \leq 2\exp\left[\mu(B)\left(\frac{\epsilon_{0}^{2}}{12}-\tau(B)\right)\right]$. The proof is concluded by substituting the bounds in the last expression.

\section{Navigability via Reducibility}\label{sec:reducibility}
 
In this section we prove our results about navigability on coherent geometries. We start by giving  a slightly more formal definition of coherence. Recall that given a geometry $(V,d)$ and a fixed (scale factor) $\gamma>1$,  
$P_{k}(v)$ denotes the number of vertices in $V$ at ``distance" $(\gamma^{k-1},\gamma^{k}]$ \blue{from $v$}. Further, for fixed $\lambda < 1$ and all  $t\neq v\in V$ , let $k_{v t}$ be the non-negative integer such that $d(v,t)\in(\gamma^{k_{vt}-1},\gamma^{k_{vt}}]$ and  $D_{\lambda}(v,t)$ be the vertices in $V$ whose distance from $v$ is at most $\gamma^{k_{vt}}$ and whose distance from $t$ is at most $\lambda\cdot  d(v,t)$. Thus, $| D_{\lambda}(v,t)|$ is the number of nodes that could facilitate greedy routing ($t$-helpful), i.e., reduce the distance to $t$ by a constant factor $\lambda<1$. 
 
\begin{coherence}\label{coherence} 
Fix $\gamma > 1$ and let $K = \lceil \log_{\gamma}(|V|)\rceil$. A geometry $(V,d)$ is $\gamma$-\emph{coherent} if:\smallskip

\noindent \emph{(H1) Bounded Growth:} $\exists A>1, \alpha>0$ such that $ P_{k}(v) \in \gamma^{k}[\alpha, A]$, for all $v \in V$ and $k \in [K]$.\smallskip

\noindent \emph{(H2) Isotropy:} $ \exists \phi>0, 1>\lambda >0$ such that $ \left|D_{\lambda}(v,t)\right| \geq  \phi\gamma^{k_{vt}}	$, for all $s\neq t \in V$.
\end{coherence}

For graphs on coherent geometries there are two requirements for navigability. The first basic requirement is deterministic and amounts to the ability to move slowly (linear rate) towards the target. In the graph augmentation setting this was given by the fact that the initial set of edges formed a connected graph. On the other hand in Kleinberg's work on set systems, the degree of vertices is set to $\Theta(\log^{2}(n))$, so that the probability of ever being stuck at a vertex is polynomially small.\red{ As mentioned in the introduction}, we opt to adopt the more natural approach of assuming a \emph{substrate}.

\red{\begin{substrate}\label{substrate}
A set of edges $E_{0}$ forms a \emph{substrate} for a geometry $(V,d)$, if for every $(s,t) \in V \times V$ with $s\neq t$, there is at least one vertex $v$ such that $\{s,v\} \in E_{0}$ and $d(v,t) \le d(s,t)-1$. If there are multiple such vertices, we distinguish one arbitrarily and call it the \emph{local $t$-connection} of $s$. A path starting from $s$ and ending to $t$ using only local $t$-connections is called a local $(s,t)$-path.
\end{substrate}}

The second requirement is probabilistic and expresses the fact that for all distance scales and ``directions" there should be significant probability of observing  an edge. This property is satisfied by Rank Based Augmentation and is essentially what was actually used to prove navigability originally.

\begin{richness} Given a $\gamma$-coherent geometry $(V,d)$ with parameters $\alpha,\phi>0$ define $k_{\theta}:= \frac{\theta \log \log n - \log a}{\log \gamma}
$ to be the distance scale of edges having distance $\Theta(\log^{\theta}(n))$. A product measure $G(n,\mathbf{Q})$ is then called $\theta$-\emph{uniformly rich} for $(V,d)$ if there is a constant $M>0$ such that for every $k\geq k_{\theta}$  every edge $(i,j)$ with $d(i,j) \in (\gamma^{k-1},\gamma^{k}]$ satisfies $Q_{ij}\geq \frac{1}{M\log^{\theta}(n)}\frac{1}{\gamma^{k}}$.
\end{richness}

In other words, since we are interested in routing in poly-logarithmic time and slow traveling can be done through the substrate (connected base graph), the probabilistic requirements concern only edges of longer distance.
As we show next these two requirements are sufficient for navigability to arise in the general setting of random graphs of bounded cost.

\subsection{Reducibility via Uniform Richness}
We start by introducing a deterministic property of graphs that implies navigability, that of \emph{reducibility}. The main advantage of  reducibility is that it allows us to separate the construction of the random graph from the analysis of the algorithm.
\begin{reducibility}
Given a graph $G(V,E)$, we will say that a pair $(s,t)\in V\times V$ is \emph{$p$-reducible} if $\exists C>0$ such that among the first $C(\log |V|)^p$ vertices of the local $(s,t)$-path there is at least one vertex $u$ such that $(u,v)\in E$ and $d(v,t)\leq \lambda d(s,t)$. If every pair $(s,t) \in V \times V$ is $p$-reducible we will say that $G$ is  \emph{$p$-reducible}.
\end{reducibility}
\begin{proposition}\label{p_reducibility}
If $G$ is $p$-reducible, greedy routing on $G$ takes at most 1+$C(\log n)^{1+p}$ steps. 
\end{proposition}
\begin{proof}[Proof of Proposition \ref{p_reducibility}]
Given any arbitrary pair of vertices $(s,t)$ with distance at most $n$, the reducibility property of $G$ guarantees us that after at most $C\log^{p}n$ steps we will obtain a new pair $(s',t)$ with distance reduced by a constant factor. Since, the new pair is also $p$-reducible, we can repeat the process until we reduce the distance again by a constant. After at most  $\log_{1/\lambda}n$ iterations we will reach the target. Since, the pairs were arbitrary, this holds for all pairs and thus the graph is navigable in 1+$C(\log n)^{1+p}$ steps. 
\end{proof}

\begin{lemma} 
\label{lem:rich} Given a $\gamma$-coherent geometry $(V,d)$ with a substrate $E_{0}$ and a random edge set $E_{q}$ sampled from a $\theta$-uniformly rich product measure $G(n,\mathbf{Q})$, the graph $G(V,E_{0}\cup E_{q})$ is $(\theta+1)$-reducible with high probability.
\end{lemma}
\begin{proof} To prove that the graph is $(\theta+1)$-reducible we will   (i) prove that the event $B_{st}$ that any fixed source-destination pair $(s,t)$ is not $	(\theta+1)$-reducible has very small probability under $G(n,\mathbf{Q})$, and (ii) use union-bound to argue that the probability that any pair is not $(\theta+1)$-reducible is small as well. 
To simplify the proof, we first 
distinguish between pairs $(s,t)$ where within the first $C\log^{\theta+1}(n)$ steps of the $t$-local path there is a vertex with distance smaller than $d(s,t)$ by a constant factor $\lambda<1$ and where there is no such vertex. Pairs $(s,t)$ that belong in the first case, are $(\theta+1)$-reducible with probability $1$. Hence, we only need to focus on the latter case, where all vertices on the first $C\log^{(\theta+1)}(n)$ steps are within the same distance scale $k_{st}:=\lceil \log_{\gamma}d(s,t)\rceil$ as $s$ from $t$. We will refer to $k_{st}$ as $k$ to ease the notation.
For each such vertex $v$ on the $t$-local path, property (H2) of coherent geometries tells us that there are at least $\phi \gamma^{k}$ candidate edges that would reduce the distance  from $t$ by a constant factor $\lambda<1$. The probability $Q_{vz}$ of each such good edge $(v,z)$ is lower bounded by $\frac{1}{M\log^{\theta+1}(n)}\frac{1}{\gamma^{k}}$, since the measure $G(n,\mathbf{Q})$ is $\theta$-uniformly rich. Let $T(s,t)$ be the set of all such good edges. We can write the probability of the event $B_{st}$ as:\vspace{-0.0cm}
\[\mathbb{P}_{\mathbf{Q}}(B_{st}) = \prod_{e\in T(s,t)} \left( 1- Q_{e}\right) \leq \left(1- \frac{1}{M\log^{\theta+1}(n) \gamma^{k}}\right)^{|T(s,t)|} \leq e^{-\frac{C\log^{\theta+1}(n)\phi \gamma^{k}}{M\log^{\theta}(n)\gamma^{k}}} \leq n^{-\frac{C\phi}{M}}
\]
where we used that $|T(s,t)|\geq C \log^{\theta+1}(n) \cdot \phi \gamma^{k}$ due to (H2) and the definition of reducibility. For any $\ell>0$ and $C\geq (2+\ell)\frac{M}{\phi}$ we get that $\P(B_{st})\leq n^{-(2+\ell)}$. To finish the proof, we perform a Union Bound over all possible sets $(s,t)$. Let $B$ be the even that the graph $G(V,E_{0}\cup E_{d})$ is not $(
\theta+1)$-reducible, then:
\[
\P_{\mathbf{Q}}(B)=\P_{\mathbf{Q}}(\bigcup B_{st}) \leq \sum_{st} \P_{\mathbf{Q}} (B_{st}) \leq n^{2} n^{-(2+\ell)} = n^{-\ell}
\]
for any $\ell>0$. Thus, the graph $G(V,E_{0}\cup E_{d})$ is $d$-navigable with high probability.
\end{proof}

\subsection{Analyzing the Product Measure $G(n,\mathbf{Q}^{*}(B))$}

Our next step will be to show that for a range of values of $B$, the product measure defined through \eqref{eq:optimizer} is $\theta$-\emph{uniformly rich} for some $\theta>0$. In doing so, our previous result shows that such a product measure leads to navigable graphs.
 Recall that  $\mathbf{Q}^{*}(B)$ is the matrix where for all $k\in[K]$ and $ij\in \mathcal{P}_{k}$ it holds that $Q^{*}_{ij}=(1+\exp(\lambda(B)c_{k}))^{-1}$ and $g(\lambda(B))$ is the expected budget corresponding to an element generated according to the product measure $\mathbf{Q}^{*}(B)$. 

\begin{proposition}\label{prop:rich} For $B\geq B^{+}_{\theta}:= \max\{B_{0}, g(\lambda_{\theta}) \}$, the product measure $G(n,\mathbf{Q}^{*}(B))$ is $\theta$-uniformly rich. The number $\lambda_{\theta}$ is explicitly defined as $ \lambda_{\theta}(\{p_{k}\},\{c_{k}\}):= \min_{k_{\theta}\leq k\leq K}\left[\frac{\log p_{k}}{c_{k}}\left(1+\frac{\theta \log\log n}{\log p_{k}}\right)\right]$.
\end{proposition}
\begin{proof}
This follows easily by the definition of $\lambda_{\theta}$. In particular, consider an edge $(i,j)$ of scale $k\geq h_{\theta}$:
\[
Q^{*}_{ij}(B)=\left[1+\exp\left(c_{k}\lambda(B)\right)\right]^{-1} \geq \left[p_{k}\log^{\theta}(n)\right]^{-1}\geq \frac{1}{A\log^{\theta}(n)\gamma^{k}}
\]
where the last inequality follows  from (H1).
\end{proof}
\begin{proposition}\label{prop:sparsity}
For $B\leq B^{-}_{\theta}:= g(\Lambda_{\theta}) $ the product measure $G(n,\mathbf{Q}^{*}(B))$ has  $O(n\cdot  \log^{\theta+1}(n))$ edges with high probability. The number $\Lambda_{\theta}$ is explicitly defined as  $ \Lambda_{\theta}(\{p_{k}\},\{c_{k}\}):= \max_{k_{\theta}\leq k\leq K}\left[\frac{\log p_{k}}{c_{k}}\left(1-\frac{\theta \log\log n}{\log p_{k}}\right)\right]$
\end{proposition}
\begin{proof}For all $B_{0}\leq B\leq B_{+}$, by definition of $\Lambda_{\theta}$ we have that for all $k\geq k_{\theta}$:
\[
Q^{*}_{ij}(B)=\left[1+\exp\left(c_{k}\lambda(B)\right)\right]^{-1} \leq \left[p_{k}\log^{-\theta}(n)\right]^{-1}
\]

Thus, the expected number of edges is upper bounded by:
\[
n\cdot \left[ Ak_{\theta}\cdot p_{k_{\theta}}+ (K-k_{\theta})\max_{k\geq k_\theta}p_{k}\frac{\log^{\theta}(n)}{p_{k}}\right]=n \cdot O\left(\log\log(n) \log^{\theta}(n)+ \log(n) \log^{\theta}(n)\right)
\]
as $k_{\theta}=O(\log\log n)$, $p_{k_{\theta}}=O(\log^{\theta}(n))$ by (H1) and $K=O(\log n)$. Applying standard Chernoff bounds~\cite{angluin1977fast} we get the required conclusion, as by definition for $B \geq B_{0}$ each class has at least a poly-logarithmic number of edges at the maximizer and thus the expected value (under the product measure) of the edges is tightly concentrated around the mean.
\end{proof}

\subsection{Analyzing Graphs of Bounded Cost}
\begin{proof}[Proof of Theorem~\ref{thm-general}]
For any $B\geq B_{0}$, consider $\mathbf{Q}^{*}(B)$ the matrix  corresponding to the optimal profile (Lemma \ref{lm:unique}) and two random elements $E^{\pm}\sim G(n,(1\pm\epsilon) \mathbf{Q}^{*}(B))$. By  Theorem   \ref{partakariola}, we get that for  $\epsilon=\sqrt{24/\log(n)}$ the probability of the event $W$, i.e. that $E^{-}\subseteq E_{\Gamma}\subseteq E^{+}$, is at least $1-n^{-5K}$. To prove Theorem \ref{thm-general} we will condition on the above event and then use our analysis of the product measure. To prove Navigability we will use the relation $E^{-}\subset E_{\Gamma}$ and the fact that Navigability is monotone property. Let $N_{d}(E)$ be the event that that the graph $G(V,E_{0}\cup E)$ is not $d$-navigable, then:
\begin{eqnarray}
\P(N_{d}(E_{c})) &=& \P(N_{d}(E_{c})\cap W) + \P(N_{d}(E_{c})\cap \bar{W})\label{in:total}\\
&\leq & \P(N_{d}(E_{c})|W) + \P( \bar{W})\label{in:bayes}\\
&\leq & \P_{\mathbf{Q}^{*}}(N_{d}(E^{-}))+ n^{-5K}\label{in:upper}\\
&\leq& n^{-\ell} + n^{-5K} \label{in:lemmas}
\end{eqnarray}
where we used the law of total probability in the first equality, Bayes Theorem and monotonicity of the probability measure in the second inequality , Theorem \ref{partakariola} and monotonicity in the third, and Lemma \ref{lem:rich} and Proposition \ref{prop:rich}  in the last. This proves part (a) of the theorem. To prove part (b) we follow the same method but for the event  $\{|E_{\Gamma}|=\omega(n \mathrm{poly}(\log(n)))\}$ and exploit the inequality $E_{\Gamma}\subset E^{+}$. Using Proposition \ref{prop:sparsity} and Theorem \ref{partakariola} we get the required conclusion.
\end{proof}

\subsection{Analysis of Indexing} 
\begin{proof}[Proof of Theorem \ref{thm-indexing}] 
We first start with the proof of part 3 of the Theorem. Instead of considering $c^{*}_{k}\propto k$ we can equivalently consider, due to (H1), $c^{*}_{k}\propto \log p_{k}$. Thus, for simplicity $c^{*}_{k}=\frac{1}{\alpha} \log p_{k}$. Set $B_{a}=g(a)$, for such $B$ and an edge $(u,v)$ of scale $k$, we have
\[
Q^{*}_{uv}  = \frac{1}{1+\exp(\lambda(B_{a})c^{*}_{k})} =   \frac{1}{1+\exp(a \frac{\log p_{k}}{a})} = \frac{1}{1+ p_{k}}
\]
Now, by property (H1) we know that for any vertex $u$ and every vertex $v$ within distance scale $k$ from $u$,  $N_{u}(d(u,v))\in [a,A]\gamma^{k}$, thus we get that:
\begin{equation}
\left(\frac{a}{2A}\right)\frac{1}{N_{u}(d(u,v))} \leq \frac{1}{2A \gamma^{k}} \leq Q^{*}_{uv}(B_{a}) \leq \frac{1}{a\gamma^{k}} \leq \left(\frac{A}{a}\right)\frac{1}{N_{u}(d(u,v))}
\end{equation}
Setting $r=2A/a$ proves part (b). To further see the correspondence between Random Graphs of Bounded Cost when the cost corresponds to indexing and Rank based augmentation, consider the $a_{k}^{*}(B_{a})$ the average number of edges of scale $k$ per vertex. We have:
\[
a_{k}^{*}(B_{a}) = \frac{p_{k}}{1+p_{k}} \approx 1, \  \forall k \in [K]
\]
Thus, we see that the scale invariance property of RBA is recovered. Furthermore, 
we have that in this case $B_{a}=\sum_{k=1}^{K}a_{k}^{*}(B_{a}) c^{*}_{k} = \Theta(\log^{2}(n))$ and the average degree of a random graph of bounded cost for $B_{a}$ is $\Theta(\log(n))$.

To show the first two parts of the theorem we essentially obtain estimates for $B^{\pm}$ given in Theorem \ref{thm-general} for the special case where the cost is the cost of indexing as above. We have:
\begin{eqnarray}
\lambda^{*}_{\theta} &=& \alpha \left(1 + \theta \frac{\log \log n}{\log p_{K}} \right) \\
\Lambda^{*}_{\theta} &=& \alpha \left(1 -\theta \frac{\log \log n}{\log p_{K}} \right)
\end{eqnarray}
By property (H1) we know that $\log p_{K} =\Theta(\log n)$. Define as before $B^{+}=g(\Lambda_{\theta}^{*})$ and $B^{-}=g(\lambda_{\theta}^{*})$ . Then for every $B^{-}\leq B\leq B^{+}$ or equivalently for $\Lambda_{\theta}\leq \lambda(B) \leq \lambda_{\theta}$, we have that for some $C>0$:
\[
 \Omega\left(\left[\log n^{-\frac{C \theta}{\log n}}\right]^{k}\right)   = \frac{p_{k}}{1+\exp(\lambda(B)c^{*}_{k})} = O\left(\left[\log n^{\frac{C \theta}{\log n}}\right]^{k}\right) 
\]
where $a_{k}^{*}(B)=\frac{p_{k}}{1+\exp(\lambda(B)c^{*}_{k})}$ expresses the average number of edges of scale $k$ per vertex. 
Thus, by (H1) we get that:
\[
B^{+}=\frac{1}{\alpha}\sum_{k=1}^{K}a_{k}^{*}(B^{+})\log p_{k} \geq \frac{1}{\alpha} \log p_{K} a_{K}^{*}(B^{+})= \Omega(\log(n)^{1+C^{'}\theta})
\]
Further, $B^{-}\leq B_{a} = \Theta(\log^{2}(n))$. Hence, we obtain that $B^{+}/B^{-} = \Omega(\mathrm{poly}(\log n))$.  The proof is concluded by invoking Theorem \ref{partakariola}.
\end{proof}

\subsection{Rank Based Augmentation for Coherent Geometries}
Recall that in RBA a single link is added for each vertex $u$ to a random vertex $v$ with  probability given by
\begin{equation}
P_{\text{\sc{rba}}}(u,v) = \frac{1}{Z} \frac{1}{\left| N_{u}\left(d(u,v)\right) \right|}
\end{equation}
where $N_{u}(\ell):=\{t \in V: d(u,t)\leq \ell\}$ is the set of vertices that are within distance $\ell$ from $u$. Here we show that the Kleinberg's original proof can be applied with ease when instead of the semi-metric induced by set-system, we have a semi-metric corresponding to a coherent geometry.  There are basically two steps. We first upper bound the normalizing constant $Z$ and then lower bound the probability that for a given pair $(s,t)$ we find an edge in the first $C\log^{2}(n)$ steps of a path along the substrate that reduces the distance to $t$ by a constant factor. 
\begin{proposition}[Bounded Growth] For a coherent geometry $(V,d)$, $\exists C<\infty$ such that $Z(1)\leq C\log(n)$.
\end{proposition}
\begin{proof}
For a given vertex $u$, we divide vertices depending on their distance scale $k\in \{0,\ldots,\log_{\gamma}(n)\}$ from $u$. For $k\geq 0$, we know from property (H1) that there are at most $A\gamma^{k}$ such vertices. Further, we also know that $|B_{k-1}(u)|=\sum_{i=0}^{k-1}P_{k}(u)\geq a \frac{\gamma^{k}-1}{\gamma-1}$. Using these two facts we have:
\[
Z(1)=\sum_{v\in V} P_{\alpha}(u,v) \leq\frac{A}{a} + \sum_{k=1}^{\log(n)} P_{k}(u) \frac{1}{|B_{k-1}(u)|} \leq\frac{A}{a} + \frac{A}{a} \sum_{k=1}^{\log(n)}\gamma^{k} \frac{\gamma-1}{\gamma^{k}-1}\leq \frac{A}{a}\left(1+ \gamma\log_{\gamma}(n)\right)
\]
\end{proof}

Finally, to complete the proof, we are going to employ once again reducibility.

\begin{proof}[Proof of Theorem \ref{thm:rba}]
Fix any two vertices $s,t$, the probability of finding a long-range edge at $s$ reducing the distance by a constant factor is at least:
\[ 
\frac{|D(s,t)|}{Z}\frac{1}{P_{k}(s)} \geq \frac{1}{C\log n}\frac{\phi \gamma^{k}}{A\gamma^{k}} = \frac{\phi}{AC} \frac{1}{\log n}
\]
Thus, the probability of the event $B_{st}$ that no such edge exists after $C'\log^{2}(n)$ trials is at most:
\[
\P(B_{st}) \leq \left(1- \frac{\phi}{AC} \frac{1}{\log n}\right)^{C'\log^{2}(n)} \leq e^{-\frac{\phi C'}{AC}\log n} \leq n^{-\frac{\phi}{AC}C'}
\]
For $C'$ large enough and a union bound over the $\Theta(n^{2})$ possible pairs of vertices, we get that if $E_{d}$ is the random set of edges added through RBA and $E_{0}$ is a substrate for the coherent geometry $(V,d)$, then the graph $G(V,E_{0} \cup E_{d})$ is $d$-navigable with high probability.
\end{proof}

\section{Set-Systems are Coherent Geometries}\label{sec:ssarecoherent}



We begin by recalling the definitions of set-systems from~\cite{NIPS}.
\begin{definition}[Set System]\label{ss:def} Let $V$ be a finite set of vertices and let $\Sigma=\{S_{1},\ldots,S_{m}\}$ be a collection of subsets of $V$. If a set $S$ contains a vertex $t$ we will say that $S$ is \emph{$t$-bound}.

Fix $0<\lambda<1$ and $\beta>1$. We say that $\Sigma$ is a $(\lambda,\beta)$-set system if all the following hold:
\begin{itemize}
\item[(K1)] $V \in \Sigma$.

\item[(K2)] If $|S|>1$, then for every $t \in S$, there is a $t$-bound  $S' \subset S$ of size $|S'| \ge \min\{\lambda |S|,|S|-1\}$.

\item[(K3)]If $S_{L}(v)$ is the union of sets that contain $v$ and have size at most $L\geq 2$, then $|S_L(v)| \leq \beta L$.
\end{itemize} 
\end{definition}

Given a set system $\Sigma$ on a set of vertices $V$, we  define the distance (semi-metric) between two vertices.



\begin{definition}
For any two vertices $u,v\in V$, their \emph{distance} in $\Sigma$, denoted by $d_{\Sigma}(u,v)$, is the size of the smallest set in $\Sigma$ containing both vertices \dima{minus 1}, i.e. $d_{\Sigma}(u,v)=\min_{S\in \Sigma}\{|S|-1: u,v\in S\}$.
\end{definition}
The goal of this section is to show that the geometry $(V,d_{\Sigma})$ is coherent for any $(\lambda,\beta)$-set system, i.e., prove that the semi-metric $d_{\Sigma}$ satisfies properties (H1) and (H2) for a suitable $\gamma>1$.
Towards that direction, the main hurdle is obtaining for all $v\in V$ upper and lower bounds on $P_{k}(v)$, the number of vertices at distance in $ (\gamma^{k-1},\gamma^{k}]$ from
$v$.  The basic observation that guides the proof is that for all $v$ and $k\geq 1$
\begin{equation}\label{eq-representation}
P_{k}(v) = |B_{k}(v)|-|B_{k-1}(v)|
\end{equation}
where $B_{k}(v)$ is the set of all vertices having distance from $v$ at most $\gamma^{k}$. 
This representation is very convenient because the properties of set systems are directly related to $|B_{k}(v)|$. In particular, if we get good upper and lower bound for $|B_{k}(v)|$ then we can obtain upper and lower bounds for $P_{k}(v)$ and prove $(H1)$, which comprises the main challenge.

Obtaining the upper bound is trivial, since it is directly given by (K3). However, the lower bound on $B_{k}(v)$ requires more thought as it needs to be tight enough so that when substituting both bounds in (4) (in order to obtain a lower bound on $P_{k}(v)$) the difference is strictly positive. It turns out that the last property depends on the particular values of the parameters $\lambda,\beta$. We show that it is always possible to select $\gamma=\gamma(\beta,\lambda)>1$ such that the last property holds. The main observation that will provide a lower bound on $|B_{k}(v)|$ is that the existence of a set $S$ with size in $(\gamma^{k-1},\gamma^{k}]$ implies that $|B_{k}(v)| \geq |S|$ for all $v\in S$. This is because all vertices in $S$ have distance at most $|S|-1$ from $v$. Thus, what remains is to show the existence of such set $S$ for all $v\in V$ and $k$. To that end, we need the following axillary lemma that was implicitly stated and used in Kleinberg's original work~\cite{NIPS}.




\begin{proposition}[Shrinkage] \label{shrinkage}
For every $S\in \Sigma$ with $|S| \ge 1/(\lambda-\lambda^2)$ and for every $t \in S$, there exists a $t$-bound set $S' \in \Sigma$ with  $\lambda^2 |S| \le |S'| \le \lambda |S|$.
\end{proposition}

\begin{proof}[Proof of Proposition~\ref{shrinkage}]
Assume, for the sake of contradiction, that there exists a set $S$ and a vertex $t \in S$ such that the proposition does not hold. If we start with $S$ and invoke (K2) iteratively until we reach $t$, we get a sequence $S=S_1 \supset S_2 \cdots \supset S_k = t$ of $t$-bound subsets of $S$. Since $|S| > \lambda |S|$, there is a largest index $i$ such that $|S_i| > \lambda |S|$, and $|S_i| \ge 2$ since $\lambda |S| > 1$. Therefore, we can apply (K2) to $S_i$ yielding  a $t$-bound set of size at least $z = \min\{\lambda |S_i|,|S_i|-1\}$. For the hypothesis to hold it must be that $z < \lambda^2 |S|$, for if $z > \lambda|S|$ we contradict the maximality of $i$. But having $z < \lambda^2 |S|$ is impossible since the fact $|S_i| > \lambda |S|$ implies $\lambda|S_i| > \lambda^2 |S|$, while combined with the fact $|S| \ge 1/(\lambda-\lambda^2)$ it implies $|S_i|-1 \ge \lambda^2 |S|$.
\end{proof}

This lemma will be used to show that for all vertices $v$ one can start from the set $V$, that belongs in $\Sigma$ by (K1), and inductively apply Lemma \ref{shrinkage} to deduce the existence of sets $S$ containing $v$ at all scales. More specifically, given a $(\lambda,\beta)$-set system $\Sigma$, let $M$ be the smallest integer such that $\lambda^{-2M} \ge |V|$. We partition the range of possible set-sizes in $\Sigma$ as  $\mathcal{I}=(I_{1},\ldots,I_{M})$ by letting $I_{k}=(\lambda^{-2(k-1)},\lambda^{-2k}]$, for $k \in [M]$. The partition $\mathcal{I}$ implicitly partitions all pairs of vertices into groups, such that all pairs in a group have roughly the same distance in $\Sigma$, i.e., up to a factor of $\lambda^{2}$.  We show that for every vertex and for every interval of the partition, there is a set with size in that interval that contains the vertex. 

\begin{proposition}\label{asetineveryscale} 
For every $t\in V$, for every $k \in [M]$, there exists a $t$-bound set $S\in \Sigma$ with $|S|\in I_{k}$.
\end{proposition}

\begin{proof}[Proof of Proposition~\ref{asetineveryscale}]
Assume, for the sake of contradiction, that there exists a vertex $t$ for which  the proposition does not hold. Let  $k_{0}\in [M]$ be the largest integer such that there is no $t$-bound set $S'\in \Sigma$ with $|S'|\in I_{k_0}$. If we start with $V$ and invoke (K2) iteratively until we reach $t$, we get a sequence $V=S_1 \supset S_2 \cdots \supset S_k = t$ of $t$-bound sets. Let  $i_{k_0}$ be the largest index $i$ such that $|S_{i}| \in I_{k_{0}+1}$. The maximality of $k_0$ implies $|S_{i_{k_{0}}+1}| \in I_{k_0 -1}$. But invoking Proposition \ref{shrinkage} for $S_{i_{k_{0}}}$ implies $|S_{i_{k_{0}}+1}|\in I_{k_{0}}$,  a contradiction.
\end{proof}

Treating $\mathcal{I}$ as a distance scale, our next goal is to obtain for each vertex $t$, upper and lower bounds on the number of vertices that lie at each distance-scale from $t$. To achieve this we need to consider a  coarser partition of the set sizes than $\mathcal{I}$. To do that it will be beneficial to use a partition built out of blocks of $\mathcal{I}$, thus allowing us to utilize Proposition~\ref{asetineveryscale}, proven for $\mathcal{I}$. In particular, the existence of a $t$-bound set of each size will be the basis for obtaining lower bounds on the number of vertices at each new distance scale from $t$. 

We let $r=r(\beta,\lambda)\ge 2$ denote the smallest integer such that $\lambda^{-2(r-1)} > \beta$ and consider the partition that results by grouping together every $r$ consecutive intervals of $\mathcal{I}$. That is, for $\gamma(\beta,\lambda) = \lambda^{-2r(\beta,\lambda)}$, we define the partition $\mathcal{A} = \mathcal{A}(\gamma)$ consisting of the intervals $A_k = (\gamma^{k-1},\gamma^{k}], \ k\in[K]$, where $K$ is the smallest integer such that $\gamma^K \ge |V|-1$. Having defined $\mathcal{A}$, we now let $P_{k}(v)$ denote the number of vertices whose distance from $v$ lies in the set $A_{k}$ and we let $P_k=\frac{1}{2}\sum_{v\in V}P_{k}(v)$ denote the total number of pairs of vertices whose distance lies in $A_k$.

\begin{lemma}[Bounded Growth] \label{boundedgrowth}
Let  $\alpha =(\lambda^{2}-\beta /\gamma)>0 $ and $ A=(\beta -\lambda^{2}/\gamma)$. For all  $k\in [K]$ and  $v\in V$,
\[
\alpha \cdot  \gamma^{k}\leq P_{k}(v)\leq A \cdot \gamma^{k} \enspace .
\]
\end{lemma}

\begin{proof}[Proof of Lemma~\ref{boundedgrowth}]
First observe that $\mathcal{A}$ is a coarsening of $\mathcal{I}$ since $\gamma = \lambda^{-2r}$ and $r\ge 2$ is an integer.  Next, let $B_{k}(v)=\sum_{i\leq k} P_{k}(v)$ be the number vertices in $V$ whose distance from $v$ lies in $A_1 \cup \cdots \cup A_k$, i.e., is no more than $\gamma^{k}$. Condition (K3) asserts that $B_{k}(v)\leq \beta \gamma^{k}$. On the other hand,  by Proposition~\ref{asetineveryscale}, we know that for any $v\in V$ there is a $v$-bound set $S\in I_{rk}\subset A_{k}$. Since, all vertices in $S$ have distance at most $|S|\leq \gamma^{k}$ from $v$, we get that $B_{k}(v)\geq |S|\ge \lambda^{-2(rk-1)}=\gamma^{k}\lambda^{2}$. Therefore, for all $k\in [K]$,
\begin{equation}\label{uplo}
\lambda^{2}\gamma^{k} \leq B_{k}(v)\leq \beta \gamma^{k} \enspace .
\end{equation}
Using the representation  \eqref{eq-representation} and invoking~\eqref{uplo}, we get
\[
\lambda^{2}\gamma^{k}-\beta \gamma^{k-1} \leq P_{k}(v) \leq \beta \gamma^{k} - \lambda^{2}\gamma^{k-1}
\]
which is equivalent to the claimed statement. The fact $\alpha > 0$ is implied by our choice of $\gamma$.
\end{proof}

\noindent Thus we have shown property (H1). Proceeding further, we need to show that the semi-metric $d_{\Sigma}$ satisfies also the isotropy property  (Section \ref{sec:reducibility}), i.e. that the size of the set  $D_{\lambda}(s,t)=\{ v\in V: d(s,v) \leq \gamma^{k_{st}} \text{ and } d(v,t)\leq \lambda d(s,t)\}$   is proportional to $\gamma^{k_{st}}$, where $k_{st}$ is the scale of $d(s,t)$. To do that we are going to show something stronger. Given any two vertices $s\neq t\in V$, consider a $S_{st}\in \Sigma$ of minimal size such that both $s,t\in S$. Then for all $k\leq k_{st}$ define the following set 
$
G_k(s,t)=\{v \in S_{st}: d(s,v)\in A_k \text{ and } d(v,t)\leq \lambda |S|  \} \enspace$
 of vertices in $S_{st}$ whose distance from $s$ lies in the interval $A_{k}$ (scale $k$) and whose distance from $t$ is no more than $\lambda |S_{st}|$.

\begin{lemma}[Isotropy]\label{success}
For every $s\neq t \in V$ with $|S_{st}| \ge 1/(\lambda-\lambda^2)$, we have that
\[
|G_{k_{st}}(s,t)\cup G_{k_{st}-1}(s,t)| \geq \left(\frac{\alpha}{\gamma}\right) \gamma^{k_{st}} \enspace .
\]
\end{lemma}

\begin{proof}[Proof of Lemma~\ref{success}] 
Proposition~\ref{shrinkage} implies that there is a $t$-bound set $S' \in \Sigma$ with  $\lambda^2 |S_{st}| \le |S'| \le \lambda |S_{st}|$. Thus, a $\lambda^2$ fraction of the vertices in $S_{st}$ have distance from $t$ at least a factor $\lambda$ less that $|S_{st}|$. Having established an abundance of  ``good" vertices in $S_{st}$, we are left to show that a constant fraction of them are in the top two distance scales $k_{st},k_{st}-1$ from $s$ (recall that $|S_{st}| \in A_{k_{st}}$). We start by noting that $Z=\sum_{i\leq k}|G_{i}(s,t)| \geq |S'|$, as the sum must count the vertices in $S'$. Since $S_{st} \in |A_{k_{st}}|$ and  $|S'| \ge \lambda^2|S_{st}|$, we get $Z \geq \lambda^2 \gamma^{k_{st}-1}$. On the other hand, the good vertices in the bottom $k_{st}-2$ distance scales from $s$ are a subset of all vertices containing $s$ at those distance scales, a quantity bounded by (K3) as  $\sum_{i\leq k-2}|G_{i}(s,t)| \leq  \beta \gamma^{k_{st}-2}$. Therefore,  $|G_{k_{st}}(s,t)\cup G_{k_{st}-1}(s,t)| \geq \lambda^{2}\gamma^{k_{st}-1}-\beta \gamma^{k_{st}-2}$. \end{proof}

\begin{proof}[Proof of Theorem \ref{coh-set}]
In order to prove that the set system defines a coherent geometry, we need to show that properties $(H1)$ and $(H2)$ hold for some $\gamma>1$. Our two lemmas achieve exactly that. The first property follows from Lemma~\ref{boundedgrowth} and the second property follows from Lemma~\ref{success} since $G_{k_{st}}(s,t)\cup G_{k_{st}-1}(s,t) \subset D_{\lambda}(s,t)$.
\end{proof}

\bibliographystyle{plain}

\end{document}